\documentclass[12pt]{article}

\usepackage[utf8]{inputenc}   
\usepackage[T1]{fontenc}      
\usepackage{lmodern}          
\usepackage{amsmath, amssymb, amsthm} 
\usepackage{mathtools}        
\usepackage{bm}               
\usepackage{bbm}              
\usepackage{mathrsfs}         
\usepackage{geometry}         
\usepackage{enumitem}         
\setlist[itemize]{leftmargin=*,topsep=2pt,itemsep=2pt}
\setlist[enumerate]{leftmargin=*,topsep=2pt,itemsep=2pt}
\usepackage{hyperref}
\usepackage{authblk}

\newcommand{\Cl}{\mathrm{Cl}}

\newcommand{\Id}{\mathrm{Id}}

\theoremstyle{definition}
\newtheorem{definition}{Definition}
\newtheorem{remark}{Remark}

\newtheorem{proposition}{Proposition}

\newtheorem{corollary}{Corollary}

\title{\textbf{From Dirac to Dunkl Operators through Symmetry Reduction}}

\author{
Cristina Sardón\thanks{\texttt{mariacristina.sardon@upm.es}} \\
Department of Applied Mathematics, Universidad Politécnica de Madrid, \\
Av.~Juan de Herrera 6, 28040 Madrid, Spain
}

\date{} 

\begin{document}

\maketitle

\begin{abstract}

This paper presents a geometric and analytic derivation of Dirac–Dunkl operators as symmetry reductions of the flat Dirac operator on Euclidean space. Starting from the standard Dirac operator, we restrict to a fundamental Weyl chamber of a finite Coxeter group equipped with the Heckman–Opdam measure, and determine the necessary drift and reflection corrections that ensure formal skew–adjointness under this weighted geometry. This procedure naturally reproduces the Dunkl operators as the unique first–order deformations compatible with reflection symmetry, whose Clifford contraction defines the Dirac–Dunkl operator and whose square yields the Dunkl Laplacian. We then extend the construction to include arbitrary unitary representations of the reflection group, obtaining representation-dependent Dirac–Dunkl operators that act on spinor- or matrix-valued functions. In the scalar and sign representations, these operators recover respectively the bosonic and fermionic Calogero–Moser systems, while higher-dimensional representations give rise to multi-component spin–Calogero models. The resulting framework unifies analytic, geometric, and representation-theoretic aspects of Dirac and Dunkl operators under a single symmetry-reduction principle.

\end{abstract}

\section{Introduction}

The Dirac operator is a paradigmatic first--order differential operator, intertwining
geometry, analysis, and representation theory. Its variants appear widely in mathematical
physics: from the Euclidean Dirac operator of spin geometry to Dunkl--type deformations
that encode reflection symmetries and inverse--square interactions. The latter, known as
\emph{Dirac--Dunkl operators}, play a central role in harmonic analysis on reflection groups,
in the theory of special functions, and in models of Calogero--Moser type \cite{Wadati1993}, \cite{Moser1975}, \cite{Calogero1971}, \cite{pere}.

In most of the literature, Dunkl operators are introduced axiomatically, defined by
combining directional derivatives with reflection terms in such a way that the resulting
family commutes \cite{Dunkl1989, DunklXu2001}. The associated Dirac--Dunkl operator is
then obtained by Clifford contraction. While this algebraic approach highlights the
commutativity and representation--theoretic structure, it leaves less transparent the
analytic origin of the reflection corrections and the role of the underlying geometry.

The aim of this paper is to derive Dirac--Dunkl operators {from first principles} \cite{heck1}, as
symmetry reductions of the free Dirac operator. Starting from the flat Dirac operator on
Euclidean space, we restrict to a fundamental Weyl chamber \cite{wall} of a finite Coxeter group \cite{Humphreys1990},\cite{Davis2008}.
equip it with the natural Heckman--Opdam measure \cite{HeckmanOpdam1987}, and determine
the necessary corrections to preserve formal skew--adjointness under integration by parts.
This procedure forces both the local ``drift'' terms and the reflection corrections,
reproducing precisely the Dunkl operators. Clifford contraction \cite{Lounesto2001} then produces the
Dirac--Dunkl operator, whose square yields the Dunkl Laplacian \cite{CastroSzarek2014}.

Hence, this paper provides a conceptual derivation of the Dirac--Dunkl operator as
the natural first--order operator compatible with reflection symmetry reduction \cite{etin1,etin2}. 

We denote by \(k : R \to \mathbb{C}\) {multiplicity function} associated with the
root system \(R\) of the reflection group \(W\), assigning a complex parameter \(k(\alpha)\)
to each root \(\alpha \in R\). Geometrically, the function \(k\) measures the strength of the
singular drift and reflection terms generated by the walls of the Weyl chambers.
When \(k=0\), these corrections vanish and one recovers the flat, translation--invariant
Dirac operator on \(\mathbb{R}^{n}\).

In the classical case ($k=0$), the operators $\partial_{\xi}$ are skew--adjoint
on the Hilbert space $L^2(\mathbb{R}^n,dx)$ (up to boundary terms), which is the
basis for Fourier analysis and the self--adjointness of the Laplacian. For $k \neq 0$, the Dunkl operators $T_{\xi}$ are no longer skew--adjoint with
respect to the standard $L^2$--inner product.  To recover the correct functional
analytic framework, one introduces the weighted inner product
\[
  \langle f,g \rangle_k \;=\; \int_{\mathbb{R}^n} f(x)\,\overline{g(x)}\,
  \delta_k(x)\,dx , \qquad
  \delta_k(x) = \prod_{\alpha \in R^+} |\langle \alpha, x \rangle|^{2k(\alpha)} .
\]
With this choice of weight, the Dunkl operators $T_{\xi}$ become skew--adjoint,
and the Dunkl Laplacian $\Delta_k$ is self--adjoint \cite{Dunkl1989}.  This is essential for the
development of Dunkl analysis: it allows the definition of a unitary Dunkl
transform, orthogonality of eigenfunctions, and the reduction of classical PDEs
to the Dunkl setting \cite{chere}, \cite{Dunkl3}, \cite{etin2}.

In \cite{CalvertDeMartino2021}, Calvert and De Martino construct a Dirac operator within 
the framework of the rational Cherednik algebra, focusing on the so-called 
{Dunkl angular momentum algebra}. Their aim is to obtain a Clifford--valued operator 
that serves as a square root of the angular Calogero--Moser Hamiltonian, the natural 
analogue of the Laplace--Dunkl operator in this setting. They introduce a Dirac element 
\(D\) built from Clifford generators tensored with Dunkl-type momentum elements, and 
show that, after the inclusion of a suitable correction term \(\phi\), the modified operator 
\(D_0 = D - \phi\) satisfies \(D_0^2 = \text{(Dunkl Laplacian)}\). This establishes a precise 
correspondence between Dirac and Laplace structures in the Dunkl context and provides 
a foundation for defining a notion of {Dirac cohomology} for representations of the 
Dunkl angular momentum algebra.
Unlike the algebraic construction in \cite{CalvertDeMartino2021}, which defines the Dirac 
operator abstractly within the rational Cherednik algebra, our present approach develops 
the {Dirac--Dunkl operator} directly as a first--order differential--difference operator 
acting on spinor--valued functions over the Euclidean space \(V\). Here, the focus lies on 
the analytic and geometric properties of the operator rather than its purely algebraic 
realization. Specifically, we start from the weighted \(L^{2}(V,\delta_{k}(x)\,dx)\) structure 
associated with the Heckman--Opdam measure and construct the Dirac--Dunkl operator 
as a deformation of the flat Dirac operator compatible with this weighted geometry. 
This formulation makes explicit the compensating drift and reflection mechanisms that 
guarantee formal skew--adjointness, providing a geometric interpretation of the 
Dunkl–Dirac correspondence beyond the algebraic setting.

\section{Fundamentals}

Let $V \simeq \mathbb{R}^{n}$ be a Euclidean space with inner product $\langle \cdot,\cdot \rangle$.
A {finite Coxeter group} $W \subset O(V)$ is a finite subgroup of the orthogonal group generated by reflections across
hyperplanes in $V$. There is a {root system} $R \subset V$ associated to $W$, consisting of
pairs $\{\pm \alpha\}$ such that each reflection
\[
s_\alpha(x) \;=\; x - \frac{2\langle \alpha, x\rangle}{\langle \alpha, \alpha\rangle} \,\alpha
\]
belongs to $W$. A choice of positive subsystem $R^{+} \subset R$ determines a
{fundamental chamber} $C \subset V$. A {multiplicity function} is a $W$--invariant map $k : R \to \mathbb{C}$, i.e.
$k(\alpha) = k(w\alpha)$ for all $w \in W$. In this space there is an associated {Heckman--Opdam weight} is
\begin{equation}\label{eq:HOweight}
  \delta_{k}(x) \;=\; \prod_{\alpha \in R^{+}} |\langle \alpha, x \rangle|^{2k(\alpha)} ,
\end{equation}
which provides the natural measure for harmonic analysis on $V$ with respect to $W$
\cite{HeckmanOpdam1987}.

\subsection{Dirac-Dunkl operators and the flat Dirac operator}
For any $\xi \in V$, the {Dunkl operator} $T_{\xi}$ is defined by
\begin{equation}\label{eq:Dunkl}
  T_{\xi} f(x) \;=\; \partial_{\xi} f(x) \;+\;
  \sum_{\alpha \in R^{+}}
  k(\alpha) \frac{\langle \alpha, \xi \rangle}{\langle \alpha, x \rangle}
  \Big( f(x) - f(s_{\alpha}x) \Big) ,
\end{equation}
where $\partial_{\xi}$ denotes the directional derivative in direction $\xi$, $\alpha$ is a root vector in $V$, $R^{+}$ is the chosen set of positive roots and $\langle \alpha, \xi \rangle$ and $\langle \alpha, x \rangle$ are inner products in $V$. This is the most general way that a Dunkl operator can be written down in any finite (Coxeter) group. 
The multiplicity function $k: R \to \mathbb{C}$ has the property that $k(w\alpha)=k(\alpha)$ for all $w\in W$.

The family $\{T_{\xi}\}_{\xi \in V}$ is commuting:
\begin{equation}\label{eq:commute}
  [T_{\xi}, T_{\eta}] = 0, \qquad \forall\, \xi, \eta \in V ,
\end{equation}
a fundamental fact due to Dunkl \cite{Dunkl1989}. The associated \emph{Dunkl Laplacian} is
\[
  \Delta_{k} \;=\; \sum_{a=1}^{n} T_{u_{a}}^{2},
\]
for any orthonormal basis $\{u_{a}\}$ of $V$.

Now, let $\{e_{a}\}_{a=1}^{n}$ be generators of the complex Clifford algebra $\Cl(V)$, satisfying
\begin{equation}\label{cliff}
  e_{a} e_{b} + e_{b} e_{a} = - 2 \delta_{ab}\Id .
\end{equation}
The flat Dirac operator on $V$ is
\begin{equation}\label{eq:DiracFlat}
  D \;=\; \sum_{a=1}^{n} e_{a} \partial_{u_{a}} ,
\end{equation}
By definition, the Dirac operator satisfies the Clifford-square identity $D^{2} = - \Delta$ ,
with $\Delta$ the standard Laplacian. This operator is formally skew-adjoint on
$L^{2}(V)$ with respect to Lebesgue measure.

We are now passing from the commuting first--order Dunkl family $\{T_{\xi}\}_{\xi\in V}$ to a Clifford--valued first--order operator by contraction with Clifford generators. Throughout,
using the defined bases $\{u_{a}\}_{a=1}^{n}$, an orthonormal basis of $V$, and 
$\{e_{a}\}_{a=1}^{n}$  the generators of the complex Clifford algebra $\Cl(V)$,
satisfying \eqref{cliff}. We write $T_{u_{a}} := T_{\xi}$ with $\xi=u_{a}$ and recall that
$[T_{u_{a}},T_{u_{b}}]=0$ \cite{Dunkl1989}.

The {Dirac--Dunkl operator} associated with a triple $(W,R,k)$ is written as:
\begin{equation}\label{eq:DiracDunkl}
  D_{k} \;=\; \sum_{a=1}^{n} e_{a}\,T_{u_{a}} .
\end{equation}
It acts on $\mathcal{S}$--valued functions, where $\mathcal{S}$ is a (finite--dimensional)
left $\Cl(V)$--module (``spinors''); by default we take $\mathcal{S}$ to be any fixed
complex spin module for $\Cl(V)$. Then, given a spinor--valued function 
\[
f : V \longrightarrow \mathcal{S}, \qquad 
x \mapsto f(x),
\] 
the operator $D_k$ acts by
\begin{equation}\label{DiracDunklop}
(D_k f)(x) \;=\; 
\sum_{a=1}^n e_a \, \big(T_{u_a} f\big)(x),
\end{equation}
that is, one first applies the Dunkl operator $T_{u_a}$ in the $u_a$--direction to each component of $f$, 
and then multiplies the result by the Clifford generator $e_a$ in the spinor module $\mathcal{S}$.

\begin{remark}[Basis independence]
If $\{u_{a}\}$ and $\{u'_{a}\}$ are orthonormal bases of $V$ with corresponding Clifford
generators $\{e_{a}\},\{e'_{a}\}$, then $D_{k}$ defined by \eqref{eq:DiracDunkl} is the
same operator: the change of basis is implemented by an orthogonal matrix and its spin
lift, and the Clifford contraction is invariant.
\end{remark}

In the following sections, we explain how the flat Dirac operator $D$ \eqref{eq:DiracFlat} can be reduced to a Dunkl-type deformation
compatible with the weighted geometry \eqref{eq:HOweight}, leading to the
Dirac--Dunkl operators \eqref{DiracDunklop}.

\section{Reduction of flat Dirac operators by finite Coxeter groups}

The passage from the flat Dirac operator to its Dunkl deformation is driven by the process of reduction by symmetry of Coxeter groups. This passage is based on two primordial steps:

\begin{itemize}
\item Restriction of the operator to the fundamental chamber.
\item Imposition of skew-adjointness of the operator with respect to the weighter measure \eqref{eq:HOweight}. 
\end{itemize}
\noindent
 In this section we analyze the procedure in detail. For it, we start with the definition of a weighted product in the reduced space.

 The weighted inner product
\begin{equation}\label{prodchamber}
  \langle f,g \rangle_k \;=\; \int_C f(x)\,\overline{g(x)}\,\delta_k(x)\,dx
\end{equation}
is defined on the Hilbert space
\[
  L^2\!\big(C,\delta_k(x)\,dx\big)
  \;=\;\left\{ f : C \to \mathbb{C} \;\;\Big|\;\;
  \int_C |f(x)|^2\,\delta_k(x)\,dx < \infty \right\}.
\]

Here $C \subset V$ is the fundamental chamber, and the measure is given by the weight
\[
  \delta_k(x) \;=\; \prod_{\alpha \in R^+} |\langle \alpha , x \rangle|^{2k(\alpha)} .
\]
For special values of \(k\), the weight reduces to familiar classical densities. For $k=0$,  we recover the plain Lebesgue measure on C, and for $k=1$ in type $A_{n-1}, \;\; \delta_{1}(x)$ is essentially the Vandermonde squared.

\begin{remark}
Since $\delta_k$ is $W$--invariant, one often extends the definition to the whole space,
and works with the Hilbert space $L^2\!\big(V,\delta_k(x)\,dx\big)$.

\end{remark}

In the unweighted case, with the standard Lebesgue measure on $V$, one has
\[
  \int f(x)\,\overline{\partial_{\xi} g(x)}\,dx
  \;=\; - \int (\partial_{\xi} f)(x)\,\overline{g(x)}\,dx ,
\]
so $\partial_{\xi}$ is skew--adjoint (up to boundary terms). In other words, this corresponds with the integration by parts identity (assuming the boundary terms vanish, e.g. for compactly supported smooth functions). With the weighted form $\langle \cdot,\cdot\rangle_k$ this fails for two reasons:
\begin{enumerate}
  \item The domain $C$ is bounded by reflection hyperplanes, and integration by parts
  produces extra boundary contributions along $\partial C$. These are called  the {wall terms} contribution along 
  \begin{equation}\label{effect1}
  \{ \langle \alpha, x \rangle = 0 \} \quad \text{for} \quad \alpha \in R^{+} 
  \end{equation}
  \item The weight $\delta_k(x)$ is not constant, and integration by parts yields an
  additional term involving
\begin{equation}\label{effect2}
    \frac{\partial_{\xi} \delta_k(x)}{\delta_k(x)}
    \;=\; \sum_{\alpha \in R^{+}} 2k(\alpha)\,
    \frac{\langle \alpha, \xi \rangle}{\langle \alpha, x \rangle}\,.
\end{equation}
 This is the {bulk drift}, derivatives fall on the weight $\delta_{k}$ and produce
a first--order term proportional to $\nabla \log \delta_{k}$.

\end{enumerate}

\noindent
These two effects (reflection and drift effect, correspondingly) are exactly compensated in the definition of the Dunkl operators,
 
\subsection{Drift compensation}
Let us start with the the drift correction first, for simplicity.
\begin{proposition}\label{prop1}
For $\xi \in V$, the adjoint of the directional derivative
$\partial_\xi$ with respect to the weighted inner product
$\langle \cdot , \cdot \rangle_k$ is given by
\[
  \partial_\xi^{*} \;=\; -\partial_\xi \;-\; \langle \nabla \log \delta_k(x), \xi \rangle ,
\]
ignoring boundary terms. In particular,
\[
  \nabla \log \delta_k(x) = \sum_{\alpha \in R^+} 
  \frac{2k(\alpha)}{\langle \alpha, x\rangle}\,\alpha .
\]

\end{proposition}

\begin{proof}
Considering the weighted inner product
\[
  \langle f,g \rangle_k
  \;=\;
  \int_C f(x)\,\overline{g(x)}\,\delta_k(x)\,dx .
\]
\noindent
Integrating by parts in this weighted measure, the derivative does not only move across, but it is picked on the weight \(\delta_k(x)\), i.e., 

\[
  \int f(x)\,\overline{\partial_\xi g(x)}\,\delta_k(x)\,dx  = -\int (\partial_\xi f)(x)\,\overline{g(x)}\,\delta_k(x)\,dx
    \;-\; \int f(x)\,\overline{g(x)}\,(\partial_\xi \delta_k(x))\,dx .
\]

Now divide the second term by \(\delta_k(x)\) on the right-hand side to write it as
\[
  - \int f(x)\,\overline{g(x)}\,
     \frac{\partial_\xi \delta_k(x)}{\delta_k(x)}\,\delta_k(x)\,dx .
\]

But
\begin{equation}\label{division}
  \frac{\partial_\xi \delta_k(x)}{\delta_k(x)}
  \;=\; \langle \nabla \log \delta_k(x), \xi \rangle.
\end{equation}

The previous expression comes from the fact that the directional derivative of a smooth function $f$ in the direction $\xi$ is defined by $\partial_\xi f(x) := \langle \nabla f(x), \xi \rangle$. Applying this to $f(x) = \log \delta_k(x)$ and using the chain rule  we retrieve \eqref{division}.

So altogether:
\[
  \langle f, \partial_\xi g \rangle_k
  \;=\;
  \langle \big(-\partial_\xi - \langle \nabla \log \delta_k(x), \xi \rangle\big) f , g \rangle_k .
\]

which is the definition of the adjoint \(\partial_\xi^*\):
\begin{equation}\label{defadjoint}
  \partial_\xi^* \;=\; -\partial_\xi \;-\; \langle \nabla \log \delta_k(x), \xi \rangle .
\end{equation}
See that one retrieves $\partial_{\xi}^{*} \;=\; -\partial_{\xi}$, the plain Lebesgue inner product, by canceling out the term on the right-hand side involving the special measure $\delta_k(x).$

\begin{remark}
The above calculation is to be understood in the sense that we ignore potential boundary contributions arising from
integration by parts.  For functions with compact support inside $C$
such terms vanish, and the formula for $\partial_\xi^{*}$ is \eqref{defadjoint}.
For more general domains of definition, one has to specify suitable
boundary conditions to obtain a rigorous operator-theoretic adjoint.
\end{remark}

We define the drift term:
\begin{equation}\label{driftdef}
  \nabla \log \delta_{k}(x)
  = \sum_{\alpha \in R^{+}} \frac{2k(\alpha)}{\langle \alpha,x\rangle}\, \alpha .
\end{equation}
\noindent
which appears from the following computation on the weight
\[
  \delta_k(x) \;=\; \prod_{\alpha \in R^+} |\langle \alpha , x \rangle|^{2k(\alpha)} .
\]

Taking logarithm on both sides of the previous expression, 
\[
  \log \delta_k(x)
  \;=\; \sum_{\alpha \in R^+} 2k(\alpha)\, \log |\langle \alpha , x \rangle| 
\]
\noindent
and further differentiation on both sides (gradient) using that \(\nabla \langle \alpha , x \rangle = \alpha\) provides:
\[
  \nabla \log \delta_k(x)
  \;=\; \sum_{\alpha \in R^+} \frac{2k(\alpha)}{\langle \alpha , x \rangle}\, \alpha .
\]

Again, with the definition of  the directional derivative
\(\partial_\xi := \langle \nabla , \xi \rangle\),  the correction term in the adjoint formula is
\[
  \langle \nabla \log \delta_k(x), \xi \rangle
  \;=\; \sum_{\alpha \in R^+} 
         \frac{2k(\alpha)\langle \alpha , \xi \rangle}{\langle \alpha , x \rangle}.
\]

\end{proof}

In addition to the bulk drift term, integration by parts may
produce contributions supported on the chamber walls
$\{\langle \alpha, x \rangle = 0\}$. Concluding, the expression would be

\begin{equation}\label{factor2}
  \int_{C} (\partial_{\xi} f)(x)\, \overline{g(x)}\, \delta_{k}(x)\, dx
  = - \int_{C} f(x)\, \overline{(\partial_{\xi} g)(x)}\, \delta_{k}(x)\, dx
    + \text{(boundary terms)}.
\end{equation}

\begin{proposition}[Drift compensator]
The bulk drift can be canceled by introducing the modified operator
\begin{equation}\label{eq:localcomp}
  \widetilde{T}_{\xi}
  \;=\; \partial_{\xi}
  \;+\; \sum_{\alpha \in R^{+}} k(\alpha)\,
       \frac{\langle \alpha, \xi \rangle}{\langle \alpha, x \rangle},
\end{equation}
which is formally skew--adjoint with
respect to $\langle \cdot , \cdot \rangle_{k}$, up to the boundary terms along
$\{\langle \alpha, x \rangle = 0\}$:
\[
  \langle \widetilde{T}_{\xi} f, g \rangle_{k}
  \;=\; - \langle f, \widetilde{T}_{\xi} g \rangle_{k}
        \;+\; \text{(boundary terms)}.
\]
\end{proposition}

\begin{proof}
Let $f,g \in C_c^\infty(C)$ be smooth functions compactly supported in the interior of $C$.
Recall that the weighted inner product is
\[
  \langle f,g \rangle_k = \int_C f(x)\,\overline{g(x)}\,\delta_k(x)\,dx,
\]
and that, by Proposition~\ref{prop1},
\[
  (\partial_\xi)^* = -\partial_\xi - \langle \nabla \log \delta_k(x), \xi \rangle
  \qquad\text{with}\qquad
  \nabla \log \delta_k(x)
  = \sum_{\alpha \in R^+} \frac{2k(\alpha)}{\langle \alpha, x\rangle}\,\alpha.
\]

\noindent
We now compute the adjoint of
\[
  \widetilde{T}_\xi f
  = \partial_\xi f
    + \sum_{\alpha \in R^+} k(\alpha)\,
      \frac{\langle \alpha,\xi\rangle}{\langle \alpha,x\rangle}\,f(x)
  = \partial_\xi f + M_\xi(x)f(x),
\]
where we denote by
\[
  M_\xi(x) := \sum_{\alpha \in R^+} k(\alpha)\,\frac{\langle \alpha,\xi\rangle}{\langle \alpha,x\rangle}
\]
the multiplication operator by the (real-valued) function $M_\xi(x)$. Since $M_\xi(x)$ is real, it is self-adjoint $(M_\xi)^* = M_\xi$. Using the formula for $(\partial_\xi)^*$ above, we obtain
\[
  (\widetilde{T}_\xi)^*
  = (\partial_\xi)^* + (M_\xi)^*
  = \big(-\partial_\xi - \langle \nabla \log \delta_k(x), \xi \rangle\big)
    + M_\xi(x).
\]

\noindent
Substituting the explicit expression for $\nabla \log \delta_k(x)$, we find
\[
  (\widetilde{T}_\xi)^*
  = -\partial_\xi
    - \sum_{\alpha \in R^+} \frac{2k(\alpha)\langle \alpha,\xi\rangle}{\langle \alpha,x\rangle}
    + \sum_{\alpha \in R^+} k(\alpha)\,\frac{\langle \alpha,\xi\rangle}{\langle \alpha,x\rangle}.
\]
The factor $2$ in the middle term originates from differentiating the weight
$\delta_k(x)$, as $\nabla \log \delta_k$ contributes a term of this form in the adjoint of $\partial_\xi$.
Combining the two sums gives
\[
  (\widetilde{T}_\xi)^*
  = -\partial_\xi
    - \sum_{\alpha \in R^+} k(\alpha)\,\frac{\langle \alpha,\xi\rangle}{\langle \alpha,x\rangle}.
\]

\noindent
The right-hand side is precisely $-\widetilde{T}_\xi$. Therefore,
\[
  (\widetilde{T}_\xi)^* = -\widetilde{T}_\xi,
\]
up to boundary terms, completing the proof.
\end{proof}

The boundary contributions arising from integration by parts cannot be removed by any
local differential correction. In the next subsection we focus on the boundary conditions arising from the integration, to which we will refer to as the reflection correction.

\subsection{Reflection correction}
The first effect that we mentioned in \eqref{effect1} to be compensated is due to the extra boundary contribution across $\partial C$.
These remaining boundary terms can be eliminated by introducing reflection operators. If we worked only inside the fundamental chamber $C$, integration by parts would leave unwanted boundary terms along the walls $\partial C$. To avoid that, we extend the integration to all of $V$, where there are no chamber boundaries. The price is that reflections appear, but since the weight $\delta_k$ is $W$--invariant, these reflection terms exactly cancel the boundary contributions \cite{Dunkl1989,HeckmanOpdam1987,Rosler2003}.

\begin{proposition}[Reflection compensator]\label{prop:ReflectionCompensator}
Let $(W,R)$ be a finite Coxeter system with multiplicity function $k$.
For each $\xi \in V$, define the Dunkl operator
\begin{equation}\label{eq:DunklDef}
  T_{\xi} f(x) \;=\; \partial_{\xi} f(x)
   + \sum_{\alpha \in R^{+}} k(\alpha)\,
     \frac{\langle \alpha, \xi \rangle}{\langle \alpha,x\rangle}
     \big( f(x) - f(s_{\alpha}x) \big) .
\end{equation}
Here $s_{\alpha}$ denotes the reflection across the wall
$\langle \alpha, x \rangle = 0$, mapping the fundamental chamber $C$
to its neighbor. The operator $T_{\xi}$ combines:
\begin{itemize}
  \item (i) a local compensator term, canceling the bulk drift produced by the weight $\delta_k$, associated with:
  \[
    \sum_{\alpha \in R^+} k(\alpha)\,\frac{\langle \alpha,\xi\rangle}{\langle \alpha,x\rangle}
  \]
  \item (ii) reflection difference term $f(x) - f(s_\alpha x)$ is responsible for canceling the wall contributions generated by integration by parts across the hyperplanes. 
\end{itemize}
\noindent
As a result, $T_{\xi}$ is formally skew--adjoint on $L^{2}(C,\delta_{k}(x)\,dx)$.
\end{proposition}

\begin{proof}

The proof consists of the two steps above: (i) the local compensator cancels the bulk drift
arising from differentiation of $\delta_{k}$, and (ii) the reflection terms cancel the
boundary contributions across the walls, since the weight $\delta_{k}$ is $W$--invariant.

Since $T_{\xi}$ is  skew--adjoint (modulo boundary terms), the following identity holds
\[
  \langle T_{\xi} f, g \rangle_{k}
  = - \langle f, T_{\xi} g \rangle_{k}, \qquad
  \forall f,g \in C_{c}^{\infty}(C).
\]

Let $f,g \in C_c^\infty(V)$ be smooth compactly supported functions. We compute
\[
  \langle T_\xi f, g \rangle_k
  = \int_V (T_\xi f)(x)\,\overline{g(x)}\,\delta_k(x)\,dx,
\]
where

\begin{equation}\label{secondterm}
  T_\xi f(x)
  = \partial_\xi f(x)
    + \underbrace{
        \sum_{\alpha \in R^+} k(\alpha)\,\frac{\langle \alpha,\xi\rangle}{\langle \alpha,x\rangle}\,
        \big( f(x) - f(s_\alpha x)\big)
      }_{\text{correction term}} .
\end{equation}

\begin{enumerate}

\item \textsl{Contribution from $\partial_\xi f$.}

By integration by parts we have
\begin{align*}
I_0=\int_V (\partial_\xi f)(x)\,\overline{g(x)}\,\delta_k(x)\,dx
&= - \int_V f(x)\,\overline{(\partial_\xi g)(x)}\,\delta_k(x)\,dx \\
&\quad - \int_V f(x)\,\overline{g(x)}\,(\partial_\xi \delta_k)(x)\,dx.
\end{align*}
Notice that here integration by parts picks up on the measure $d\mu_k(x) = \delta_k(x)\,dx.$ Using the definition of the drift given in \eqref{driftdef}, the previous integration by parts becomes
\begin{align}\label{eq:adjoint-derivative}
I_0=\int_V (\partial_\xi f)(x)\,\overline{g(x)}\,\delta_k(x)\,dx
&= - \int_V f(x)\,\overline{(\partial_\xi g)(x)}\,\delta_k(x)\,dx \nonumber\\
 & - \sum_{\alpha \in R^+}\int_V
    f(x)\,\overline{g(x)}\,\frac{2k(\alpha)\langle \alpha,\xi\rangle}{\langle \alpha,x\rangle}\,\delta_k(x)\,dx.
\end{align}

\item \textsl{Contribution from the correction term.}  
Now consider
\[
  \sum_{\alpha \in R^+} k(\alpha)\int_V
    \frac{\langle \alpha,\xi\rangle}{\langle \alpha,x\rangle}\,
      \big( f(x) - f(s_\alpha x)\big)\,
      \overline{g(x)}\,\delta_k(x)\,dx.
\]

Split this into two parts:
\[
  I_1 = \sum_{\alpha \in R^+} k(\alpha)\int_V
    \frac{\langle \alpha,\xi\rangle}{\langle \alpha,x\rangle}\,
      f(x)\,\overline{g(x)}\,\delta_k(x)\,dx,
\]
and
\[
  I_2 = -\sum_{\alpha \in R^+} k(\alpha)\int_V
    \frac{\langle \alpha,\xi\rangle}{\langle \alpha,x\rangle}\,
      f(s_\alpha x)\,\overline{g(x)}\,\delta_k(x)\,dx.
\]

\medskip

\noindent
\item \textsl{Simplify $I_2$ using $W$--invariance.}  

Under the change of variables $y = s_\alpha x$, and that the reflection 
\begin{equation}\label{reflection1}
s_{\alpha}(x) = x - \frac{2\langle \alpha, x \rangle}{\langle \alpha, \alpha \rangle}\,\alpha,
\end{equation}
is an orthogonal linear map with $\det(s_\alpha) = -1$, we have $dx = dy$. Moreover,
since $\delta_k$ is $W$--invariant, one has $\delta_k(s_\alpha y) = \delta_k(y)$.
Therefore, the weighted measure $\delta_k(x)\,dx$ is invariant under $s_\alpha$.

Thus,
\begin{equation}\label{integral2}
I_2 = -\sum_{\alpha \in R^+} k(\alpha)\int_V
    \frac{\langle \alpha,\xi\rangle}{\langle \alpha,s_\alpha y\rangle}\,
      f(y)\,\overline{g(s_\alpha y)}\,\delta_k(y)\,dy.
\end{equation}
\noindent
Now use the fact that $s_\alpha$ acts on $\alpha$ by $s_\alpha(\alpha) = -\alpha$, which can be trivially seen by setting $\alpha=x$ in \eqref{reflection1}, and now applying $s_\alpha$ to a generic vector $y \in V$:
\[
s_\alpha(y) = y - \frac{2\langle \alpha, y \rangle}{\langle \alpha, \alpha \rangle}\,\alpha.
\]
and taking the inner product of this with $\alpha$:
\[
\langle \alpha, s_\alpha y \rangle
= \langle \alpha, y \rangle
  - \frac{2\langle \alpha, y \rangle}{\langle \alpha, \alpha \rangle}\langle \alpha, \alpha \rangle
= \langle \alpha, y \rangle - 2\langle \alpha, y \rangle
= -\langle \alpha, y \rangle,
\]
we obtain
\[
\langle \alpha, s_\alpha y \rangle = -\langle \alpha, y \rangle.
\]

Therefore, we can transform the integrand in \eqref{integral2} as
\[
  \frac{\langle \alpha,\xi\rangle}{\langle \alpha,s_\alpha y\rangle}
  = -\,\frac{\langle \alpha,\xi\rangle}{\langle \alpha,y\rangle}.
\]

and follows that
\[
  I_2 = \sum_{\alpha \in R^+} k(\alpha)\int_V
    \frac{\langle \alpha,\xi\rangle}{\langle \alpha,y\rangle}\,
      f(y)\,\overline{g(s_\alpha y)}\,\delta_k(y)\,dy.
\]

\item \textsl{Combine all contributions.}  

Gathering the results from the previous steps, we can express
\[
  \langle T_\xi f, g \rangle_k
  = I_0+I_1+I_2
\]

where ssing \eqref{eq:adjoint-derivative}, \eqref{integral2}, and the simplifications above, these are explicitly:
\begin{align*}
I_0
  &= -\int_V f(x)\,\overline{(\partial_\xi g)(x)}\,\delta_k(x)\,dx
     - \sum_{\alpha \in R^+}\int_V
        f(x)\,\overline{g(x)}\,
        \frac{2k(\alpha)\langle \alpha,\xi\rangle}{\langle \alpha,x\rangle}\,
        \delta_k(x)\,dx,\\[3mm]
I_1
  &= \sum_{\alpha \in R^+}\int_V
        f(x)\,\overline{g(x)}\,
        k(\alpha)\frac{\langle \alpha,\xi\rangle}{\langle \alpha,x\rangle}\,
        \delta_k(x)\,dx,\\[3mm]
I_2
  &= \sum_{\alpha \in R^+}\int_V
        f(x)\,\overline{g(s_\alpha x)}\,
        k(\alpha)\frac{\langle \alpha,\xi\rangle}{\langle \alpha,x\rangle}\,
        \delta_k(x)\,dx.
\end{align*}

Now sum we sum these three integrals.
The first term in $I_0$ already matches the expected derivative term
$-\partial_\xi g$ in the adjoint.  
The second term in $I_0$ combines with $I_1$:
\[
-2k(\alpha) + k(\alpha) = -k(\alpha),
\]
so that the remaining bulk term is
\[
 -\sum_{\alpha \in R^+} k(\alpha)
   \frac{\langle \alpha,\xi\rangle}{\langle \alpha,x\rangle}
   f(x)\,\overline{g(x)}\,\delta_k(x).
\]
Finally, $I_2$ contributes the reflection integral
\[
 \sum_{\alpha \in R^+} k(\alpha)\int_V
   \frac{\langle \alpha,\xi\rangle}{\langle \alpha,x\rangle}
   f(x)\,\overline{g(s_\alpha x)}\,\delta_k(x)\,dx.
\]

Combining everything, we arrive at
\begin{align*}
\langle T_\xi f, g \rangle_k
  &= -\int_V f(x)\,\overline{(\partial_\xi g)(x)}\,\delta_k(x)\,dx \\
  &\quad - \sum_{\alpha \in R^+} k(\alpha)\int_V
       \frac{\langle \alpha,\xi\rangle}{\langle \alpha,x\rangle}\,
       f(x)\,\overline{g(x)}\,\delta_k(x)\,dx \\
  &\quad + \sum_{\alpha \in R^+} k(\alpha)\int_V
       \frac{\langle \alpha,\xi\rangle}{\langle \alpha,x\rangle}\,
       f(x)\,\overline{g(s_\alpha x)}\,\delta_k(x)\,dx,
\end{align*}
which coincides with $-\langle f, T_\xi g\rangle_k$.

\end{enumerate}
\end{proof}

\medskip

\noindent
We conclude that the Dunkl operators \eqref{eq:DunklDef} constitute the natural deformation of the directional derivatives $\partial_\xi$ compatible with both the bulk structure induced by $\delta_k$ and the reflection symmetries of the Coxeter group $W$.

\section{Dirac--Dunkl construction}

In analogy with the flat Dirac operator $D = \sum_{a=1}^n e_a \partial_{u_a}$,
we define its Dunkl deformation by replacing the derivatives $\partial_{u_a}$
with Dunkl operators $T_{u_a}$.

\begin{definition}[Dirac--Dunkl operator]\label{def:DiracDunkl}
Let $\{u_a\}_{a=1}^n$ be an orthonormal basis of $V$, and let
$\{e_a\}_{a=1}^n$ denote the corresponding Clifford generators acting
on the spinor space $\mathcal{S}$. The Dirac--Dunkl operator is
\begin{equation}
  D_k \;=\; \sum_{a=1}^n e_a \, T_{u_a},
\end{equation}\label{DDop}
where $T_{u_a}$ are the Dunkl operators from \eqref{eq:DunklDef}.
\end{definition}

\begin{proposition}[Formal skew--adjointness]\label{prop:DkSkew}
With respect to the weighted inner product on a fundamental chamber
\[
  \langle f,g\rangle_{k}
  \;=\;\int_{C}\!\langle f(x),g(x)\rangle_{\mathcal{S}}\,
         \delta_{k}(x)\,dx ,
\]
the operator $D_{k}$ is formally skew--adjoint on
$C_{c}^{\infty}(C;\mathcal{S})$:
\[
  \langle D_{k} f, g\rangle_{k} \;=\; -\,\langle f, D_{k} g\rangle_{k}.
\]
\end{proposition}

\begin{proof}
By Proposition~\ref{prop:ReflectionCompensator}, each $T_{u_{a}}$ is formally skew--adjoint on
$C_{c}^{\infty}(C)$ with respect to $\langle\cdot,\cdot\rangle_{k}$; the Clifford
generators $e_{a}$ act unitarily on $\mathcal{S}$ with
$e_{a}^{\dagger}=-\,e_{a}$, so that is formally skew--adjoint
wrt to the weighted  space associated to the Heckman–Opdam measure \eqref{cliff}. Thus
\[
  \langle e_{a}T_{u_{a}} f, g\rangle_{k}
  \;=\; \langle T_{u_{a}} f, e_{a}^{\dagger} g\rangle_{k}
  \;=\; -\,\langle T_{u_{a}} f, e_{a} g\rangle_{k}
  \;=\; \langle f, T_{u_{a}}(e_{a} g)\rangle_{k}
  \;=\; -\,\langle f, e_{a} T_{u_{a}} g\rangle_{k}.
\]
Summing over $a$ yields the claim.
\end{proof}

\begin{proposition}[Dirac square identity]\label{thm:DiracSquare}
For the Dirac--Dunkl operator \eqref{def:DiracDunkl} one has the Clifford square identity
\begin{equation}\label{eq:DkSquare}
  D_{k}^{2} \;=\; -\,\Delta_{k},\qquad
  \Delta_{k} \;=\; \sum_{a=1}^{n} T_{u_{a}}^{2}.
\end{equation}
\end{proposition}

\begin{proof}
Expand using \eqref{def:DiracDunkl}:
\[
  D_{k}^{2}
  \;=\; \sum_{a,b} e_{a} e_{b}\, T_{u_{a}} T_{u_{b}}
  \;=\; \tfrac{1}{2}\sum_{a,b} (e_{a} e_{b}+e_{b} e_{a})\, T_{u_{a}} T_{u_{b}}
        \;+\; \tfrac{1}{2}\sum_{a,b} (e_{a} e_{b}-e_{b} e_{a})\, T_{u_{a}} T_{u_{b}}.
\]
The commutator part vanishes because $[T_{u_{a}},T_{u_{b}}]=0$; hence only the
anti-commutator of the Clifford generators remains. Using
$e_{a}e_{b}+e_{b}e_{a}=-2\delta_{ab} I$ gives
\[
  D_{k}^{2} \;=\; -\sum_{a} T_{u_{a}}^{2} \;=\; -\,\Delta_{k}.
\]
\end{proof}

\begin{corollary}
In the flat case $k\equiv 0$ one recovers $D^{2}=-\Delta$, the usual Dirac square identity.
\end{corollary}

\subsection{Global $W$--equivariant formulation}

In the previous section, the integration over $V$ was used as a
{technical device} to eliminate boundary terms arising from
integration by parts inside the chamber $C$. In that setting, fields
were originally defined on $C$, and extended to $V$ only for the sake
of the adjointness calculation.

We now pass to a genuinely global formulation, where fields are defined
on all of $V$ from the start, subject to a $W$--equivariance rule.
Let $\rho: W \to U(\mathcal{S})$ be a finite--dimensional unitary
representation, commuting with the Clifford action \cite{Lounesto2001}. A
$\mathcal{S}$--valued field $\Phi$ on $C$ is extended to all of $V$ by
the glue rule
\begin{equation}\label{eq:GlueRuleGlobal}
  \Phi(w x) \;=\; \rho(w)\,\Phi(x),
  \qquad x\in C,\; w\in W.
\end{equation}
This provides a global $W$--equivariant field on $V$, consistent across
the chamber walls. With this convention, the Dunkl operators acting on $\mathcal{S}$--valued
fields take the global difference form
\begin{equation}\label{eq:DunklRho}
  T_{\xi} \Phi(x) \;=\; \partial_{\xi}\Phi(x)
  \;+\; \sum_{\alpha\in R^{+}} k(\alpha)\,
        \frac{\langle \alpha,\xi\rangle}{\langle \alpha,x\rangle}
        \Big(\Phi(x)-\rho(s_{\alpha})\,\Phi(s_{\alpha}x)\Big),
\end{equation}
and the Dirac--Dunkl operator is given by
\begin{equation}\label{eq:DiracDunklGlobal}
  D_{k} \;=\; \sum_{a=1}^{n} e_{a}\, T_{u_{a}}.
\end{equation}

\begin{remark}
In the scalar case $\rho=\mathbf{1}$ the formula
\eqref{eq:DunklRho} reduces to the reflection-compensated Dunkl operator
\eqref{eq:DunklDef}. For nontrivial $\rho$, the reflection terms act by
$\rho(s_\alpha)$ on the internal spin space $\mathcal{S}$, twisting the
operator in a representation--theoretic manner \cite{VazRocha2016}.
\end{remark}

\section{Dirac-Dunkl Operators with Nontrivial Representations}

To understand how representation dependence modifies the structure of Dunkl and Dirac--Dunkl operators, we begin with the simplest Coxeter systems, namely of type $A_1$ and $A_2$.  These low-rank examples make explicit the interplay between the reflection group $W$, its unitary representations $\rho$, and the corresponding Dunkl operators.  Starting from these cases allows us to track how the reflection terms, differential parts, and rational factors evolve when the representation changes—from the trivial and sign representations to higher-dimensional irreducible ones.  The resulting operators exhibit different symmetry properties and lead to distinct rational (Calogero-type) reductions.  Once these basic mechanisms are clear, the generalization to arbitrary type $A_n$ follows naturally.

For simplicity, let us with the one dimensional case. In this case, $V=\mathbb{R}$, the reflection group is $W=\mathbb{Z}_2=\{1,s\}$,
with reflection $s:x\mapsto -x$, and root system $R^+=\{1\}$.
The weight is $\delta_k(x)=|x|^{2k}$.
The Dunkl operator with representation $\rho$ is
\[
  T f(x) \;=\; \frac{d}{dx} f(x)
  \;+\; k\,\frac{f(x)-\rho(s)\,f(-x)}{x}.
\]

\begin{itemize}
\item[(i)] Trivial representation $\rho=\mathbf{1}$.
In this case,
\[
  T f(x) \;=\; f'(x) + k\,\frac{f(x)-f(-x)}{x},
\]
the classical one-dimensional Dunkl operator. The Dirac--Dunkl operator is
\[
  D_k f(x) = e_1\,T f(x),
\]
with $D_k^2 = -T^2$.

\item[(ii)] Sign representation $\rho(s)=-1$.
Then
\[
  T f(x) \;=\; f'(x) + k\,\frac{f(x)-(-f(-x))}{x}
  \;=\; f'(x) + k\,\frac{f(x)+f(-x)}{x}.
\]
Here the reflection term flips sign, producing a genuinely different operator.
Again $D_k = e_1 T$, but with this twisted form of $T$.
\end{itemize}
\noindent
Let us picture now $A_2$ so we can generalize afterwards to the case $A_n$.
Now $V=\{(x_1,x_2,x_3)\in\mathbb{R}^3 \;|\; x_1+x_2+x_3=0\}\cong \mathbb{R}^2$,
with $W=S_3$ acting by permutation of coordinates.
The root system has three positive roots, e.g.
\[
  R^+ = \{e_1-e_2, \; e_2-e_3, \; e_1-e_3\},
\]
and reflections across the hyperplanes
$\{x_i=x_j\}$. For $\xi\in V$, the Dunkl operator with representation $\rho$ is
\[
  T_\xi \Phi(x)
  = \partial_\xi \Phi(x)
  + \sum_{\alpha\in R^+} k(\alpha)\,
    \frac{\langle \alpha,\xi\rangle}{\langle \alpha,x\rangle}
    \big(\Phi(x)-\rho(s_\alpha)\,\Phi(s_\alpha x)\big).
\]

\begin{itemize}
\item[(i)] Trivial representation $\rho=\mathbf{1}$.

In this case, $\rho(s_\alpha) = \Id$ for all reflections $s_\alpha \in W$, so the
operator acts on scalar functions $\Phi : V \to \mathbb{C}$. The reflection term
simplifies to the scalar difference
\[
  \Phi(x) - \Phi(s_\alpha x),
\]
measuring the jump of $\Phi$ across the reflecting hyperplane
$\langle \alpha, x \rangle = 0$. Hence, the Dunkl operator becomes
\begin{equation}\label{DDa2gen}
  T_\xi \Phi(x)
  = \partial_\xi \Phi(x)
  + \sum_{\alpha\in R^+} k(\alpha)\,
    \frac{\langle \alpha,\xi\rangle}{\langle \alpha,x\rangle}
    \big(\Phi(x)-\Phi(s_\alpha x)\big),
\end{equation}
which is precisely the classical Dunkl operator of type $A_2$ acting on scalar
functions.

Geometrically, $A_2$ corresponds to reflections across the three lines
$x_1=x_2$, $x_2=x_3$, and $x_1=x_3$ in the plane $x_1+x_2+x_3=0$, forming $6$
Weyl chambers separated by $60^\circ$ angles. The scalar Dunkl operator therefore
encodes infinitesimal translations within one chamber, together with discrete
reflections across the walls. The parameter $k(\alpha)$ controls the strength of
these reflections: when $k=0$, $T_\xi$ reduces to the standard directional derivative,
while for $k\neq 0$ it introduces a nonlocal correction capturing the influence
of neighboring chambers through $\Phi(s_\alpha x)$.
To rewrite \eqref{DDa2gen} in $\mathbb{R}^2,$ we consider the ambient space:
\begin{equation}\label{Va2}
  V = \{(x_1,x_2,x_3)\in\mathbb{R}^3 \;|\; x_1+x_2+x_3=0\}
  \;\cong\; \mathbb{R}^2,
\end{equation}
and the reflection group $W=S_3$ acts by permuting the coordinates $(x_1,x_2,x_3)$.

The three positive roots may be chosen as
\[
  \alpha_1 = e_1 - e_2, \qquad
  \alpha_2 = e_2 - e_3, \qquad
  \alpha_3 = e_1 - e_3,
\]
corresponding to reflections across the planes
$x_1 = x_2$, $x_2 = x_3$, and $x_1 = x_3$ respectively.  The associated reflections act by
\[
  s_{\alpha_1}(x_1,x_2,x_3) = (x_2,x_1,x_3), \qquad
  s_{\alpha_2}(x_1,x_2,x_3) = (x_1,x_3,x_2), \qquad
  s_{\alpha_3}(x_1,x_2,x_3) = (x_3,x_2,x_1).
\]

Therefore, the Dunkl operator along a direction $\xi=(\xi_1,\xi_2,\xi_3)\in V$ acts as
\[
  (T_\xi \Phi)(x_1,x_2,x_3)
  = \sum_{i=1}^3 \xi_i\,\partial_{x_i}\Phi(x_1,x_2,x_3)
  + \sum_{\alpha\in R^+} k(\alpha)\,
    \frac{\langle \alpha,\xi\rangle}{\langle \alpha,x\rangle}
    \Big(\Phi(x_1,x_2,x_3)-\Phi(s_\alpha(x_1,x_2,x_3))\Big),
\]
where $(x_1,x_2,x_3)$ fulfill the conditions \eqref{Va2}. Writing out the scalar products explicitly, we have
\[
\begin{aligned}
  \langle \alpha_1,\xi\rangle &= \xi_1-\xi_2, & \quad
  \langle \alpha_1,x\rangle &= x_1-x_2,\\
  \langle \alpha_2,\xi\rangle &= \xi_2-\xi_3, & \quad
  \langle \alpha_2,x\rangle &= x_2-x_3,\\
  \langle \alpha_3,\xi\rangle &= \xi_1-\xi_3, & \quad
  \langle \alpha_3,x\rangle &= x_1-x_3.
\end{aligned}
\]

Hence, the Dunkl operator in coordinates becomes
\[
\begin{aligned}
  (T_\xi \Phi) & (x_1,x_2,x_3)
  = \xi_1\,\partial_{x_1}\Phi
   + \xi_2\,\partial_{x_2}\Phi
   + \xi_3\,\partial_{x_3}\Phi \\
  & + k\Bigg[
     \frac{(\xi_1-\xi_2)\big(\Phi(x_1,x_2,x_3)-\Phi(x_2,x_1,x_3)\big)}{x_1-x_2}
     + \frac{(\xi_2-\xi_3)\big(\Phi(x_1,x_2,x_3)-\Phi(x_1,x_3,x_2)\big)}{x_2-x_3}\\
  &\qquad\quad \qquad\quad
     + \frac{(\xi_1-\xi_3)\big(\Phi(x_1,x_2,x_3)-\Phi(x_3,x_2,x_1)\big)}{x_1-x_3}
   \Bigg].
\end{aligned}
\]

This is the explicit classical Dunkl operator of type $A_2$ in the scalar case.  
Geometrically, the three hyperplanes $x_i = x_j$ divide the plane
$x_1+x_2+x_3=0$ into six $60^\circ$ chambers. The Dunkl term introduces
nonlocal reflection corrections coupling adjacent chambers, with coupling
strength $k(\alpha)$. For $k=0$ one recovers the flat derivative in direction $\xi$.

\item[(ii)] \text{Sign representation $\rho = \mathrm{sgn}$ of $S_3$.}

Here each reflection acts as multiplication by $-1$, so that
$\rho(s_\alpha)\Phi(s_\alpha x) = -\,\Phi(s_\alpha x)$. The Dunkl operator becomes
\[
  T_\xi \Phi(x)
  = \partial_\xi \Phi(x)
  + \sum_{\alpha\in R^+} k(\alpha)
     \frac{\langle \alpha,\xi\rangle}{\langle \alpha,x\rangle}
     \big(\Phi(x) + \Phi(s_\alpha x)\big), \qquad x\in V.
\]
or explicitly, in coordinates:
\[
\begin{aligned}
T_\xi \Phi(x_1,x_2,x_3)
  &= \sum_{i=1}^3 \xi_i\,\partial_{x_i}\Phi(x_1,x_2,x_3) \\
  &\quad+ k\Bigg[
     \frac{(\xi_1-\xi_2)\big(\Phi(x)+\Phi(x_2,x_1,x_3)\big)}{x_1-x_2}
     + \frac{(\xi_2-\xi_3)\big(\Phi(x)+\Phi(x_1,x_3,x_2)\big)}{x_2-x_3}\\
  &\qquad\quad
     + \frac{(\xi_1-\xi_3)\big(\Phi(x)+\Phi(x_3,x_2,x_1)\big)}{x_1-x_3}
     \Bigg].
\end{aligned}
\]

Recall that $x\in V$ defined in \eqref{Va2}, as well as $(x_1,x_2,x_3)\in V$, i.e., $x_3=-x_1-x_2$.
This differs from the scalar (trivial) case by a sign flip in the reflection term,
coupling even and odd parts of $\Phi$. In this antisymmetric (sign) sector, the
square of the Dirac--Dunkl operator gives rise to a Calogero--Moser type Hamiltonian
with inverse--square wall potentials.

\item[(iii)] Two-dimensional irreducible representation of $S_3$.
Here $\Phi:V\to\mathbb{C}^2$ with $\Phi(x)=\begin{pmatrix}\phi_1(x) \\ \phi_2(x)\end{pmatrix}$,
and the reflections act by $2\times 2$ orthogonal matrices:
\[
\rho(s_{12})=\begin{pmatrix}1&0\\0&-1\end{pmatrix},\quad
\rho(s_{23})=\begin{pmatrix}-\tfrac{1}{2}&\tfrac{\sqrt{3}}{2}\\[6pt]\tfrac{\sqrt{3}}{2}&\tfrac{1}{2}\end{pmatrix},\quad
\rho(s_{13})=\begin{pmatrix}-\tfrac{1}{2}&-\tfrac{\sqrt{3}}{2}\\[6pt]-\tfrac{\sqrt{3}}{2}&\tfrac{1}{2}\end{pmatrix}.
\]

Thus the Dunkl operator is
\[
\begin{aligned}
T_\xi \Phi(x) &= \partial_\xi \Phi(x) \\[6pt]
&\quad+ k(\alpha_{12}) \frac{\langle \alpha_{12},\xi\rangle}{\langle \alpha_{12},x\rangle}
\left(
\begin{pmatrix}\phi_1(x) \\ \phi_2(x)\end{pmatrix}
-
\begin{pmatrix}1&0\\0&-1\end{pmatrix}
\begin{pmatrix}\phi_1(x_2,x_1,x_3) \\ \phi_2(x_2,x_1,x_3)\end{pmatrix}
\right) \\[12pt]
&\quad+ k(\alpha_{23}) \frac{\langle \alpha_{23},\xi\rangle}{\langle \alpha_{23},x\rangle}
\left(
\begin{pmatrix}\phi_1(x) \\ \phi_2(x)\end{pmatrix}
-
\begin{pmatrix}-\tfrac{1}{2}&\tfrac{\sqrt{3}}{2}\\[6pt]\tfrac{\sqrt{3}}{2}&\tfrac{1}{2}\end{pmatrix}
\begin{pmatrix}\phi_1(x_1,x_3,x_2) \\ \phi_2(x_1,x_3,x_2)\end{pmatrix}
\right) \\[12pt]
&\quad+ k(\alpha_{13}) \frac{\langle \alpha_{13},\xi\rangle}{\langle \alpha_{13},x\rangle}
\left(
\begin{pmatrix}\phi_1(x) \\ \phi_2(x)\end{pmatrix}
-
\begin{pmatrix}-\tfrac{1}{2}&-\tfrac{\sqrt{3}}{2}\\[6pt]-\tfrac{\sqrt{3}}{2}&\tfrac{1}{2}\end{pmatrix}
\begin{pmatrix}\phi_1(x_3,x_2,x_1) \\ \phi_2(x_3,x_2,x_1)\end{pmatrix}
\right), 
\end{aligned}
\]
such that $x\in V, (x_1,x_2,x_3)\in V.$
Explicitly, each reflection term gives:
\[
\begin{aligned}
\Phi(x) - \rho(s_{12})\Phi(s_{12}x)
&=
\begin{pmatrix}
\phi_1(x_1,x_2,x_3) - \phi_1(x_2,x_1,x_3) \\[4pt]
\phi_2(x) + \phi_2(x_2,x_1,x_3)
\end{pmatrix}, \\[8pt]
\Phi(x) - \rho(s_{23})\Phi(s_{23}x)
&=
\begin{pmatrix}
\phi_1(x) + \tfrac{1}{2}\phi_1(x_1,x_3,x_2) - \tfrac{\sqrt{3}}{2}\phi_2(x_1,x_3,x_2) \\[6pt]
\phi_2(x) - \tfrac{\sqrt{3}}{2}\phi_1(x_1,x_3,x_2) - \tfrac{1}{2}\phi_2(x_1,x_3,x_2)
\end{pmatrix}, \\[8pt]
\Phi(x) - \rho(s_{13})\Phi(s_{13}x)
&=
\begin{pmatrix}
\phi_1(x) + \tfrac{1}{2}\phi_1(x_3,x_2,x_1) + \tfrac{\sqrt{3}}{2}\phi_2(x_3,x_2,x_1) \\[6pt]
\phi_2(x) + \tfrac{\sqrt{3}}{2}\phi_1(x_3,x_2,x_1) - \tfrac{1}{2}\phi_2(x_3,x_2,x_1)
\end{pmatrix}.
\end{aligned}
\]

Hence $T_\xi$ acts on $\Phi:V\to\mathbb{C}^2$ by combining derivatives with explicit
linear combinations of reflected values of the components, with coefficients given by the
reflection matrices of the $2$-dimensional irrep of $S_3$.
\end{itemize}´

\subsection{Explicit Calogero--type realizations for type $A_2$.}
Let $\xi = e_1 - e_2 = \alpha_{12}$, so that
$\partial_\xi = \partial_{x_1} - \partial_{x_2}$.
The positive roots are
\[
\alpha_{12}=e_1-e_2,\qquad
\alpha_{23}=e_2-e_3,\qquad
\alpha_{13}=e_1-e_3,
\]
with
\[
\langle \alpha_{12},\xi\rangle=2,\qquad
\langle \alpha_{23},\xi\rangle=-1,\qquad
\langle \alpha_{13},\xi\rangle=1,
\]
and for any $x=(x_1,x_2,x_3)$,
\[
\langle \alpha_{12},x\rangle=x_1-x_2,\qquad
\langle \alpha_{23},x\rangle=x_2-x_3,\qquad
\langle \alpha_{13},x\rangle=x_1-x_3.
\]
Then the general Dunkl operator is
\[
\begin{aligned}
T_{e_1-e_2}\Phi(x)
&= \partial_{x_1}\Phi(x_1,x_2,x_3)-\partial_{x_2}\Phi(x_1,x_2,x_3) \\[3pt]
&\quad+ k(\alpha_{12})\,\frac{2}{x_1-x_2}
\big(\Phi(x)-\rho(s_{12})\Phi(x_2,x_1,x_3)\big) \\[3pt]
&\quad- k(\alpha_{23})\,\frac{1}{x_2-x_3}
\big(\Phi(x)-\rho(s_{23})\Phi(x_1,x_3,x_2)\big) \\[3pt]
&\quad+ k(\alpha_{13})\,\frac{1}{x_1-x_3}
\big(\Phi(x)-\rho(s_{13})\Phi(x_3,x_2,x_1)\big).
\end{aligned}
\]

\medskip

\noindent\textbf{(i) Trivial representation \(\rho=\mathbf{1}\).}\\[2pt]
Here $\rho(s_{ij})=1$, so all reflections act trivially and
\begin{align}
T_{e_1-e_2}\Phi(x_1,x_2,x_3)
&= (\partial_{x_1}-\partial_{x_2})\Phi(x_1,x_2,x_3)+ 2k_{12}\frac{\Phi(x_1,x_2,x_3)-\Phi(x_2,x_1,x_3)}{x_1-x_2} \nonumber \\
&\quad \quad  - k_{23}\frac{\Phi(x_1,x_2,x_3)-\Phi(x_1,x_3,x_2)}{x_2-x_3}+ k_{13}\frac{\Phi(x_1,x_2,x_3)-\Phi(x_3,x_2,x_1)}{x_1-x_3}.
\end{align}
This reproduces the usual scalar Dunkl operator of type $A_2$, i.e.
the standard Calogero--Moser differential--difference form.

\noindent\textbf{(ii) Sign representation \(\rho=\mathrm{sgn}\).}\\[2pt]
Now $\rho(s_{ij})=-1$, so each reflection contributes with a plus sign:
\begin{align}
T_{e_1-e_2}\Phi(x_1,x_2,x_3)
&= (\partial_{x_1}-\partial_{x_2})\Phi(x_1,x_2,x_3)
+ 2k_{12}\frac{\Phi(x_1,x_2,x_3)+\Phi(x_2,x_1,x_3)}{x_1-x_2} \nonumber\\
&- k_{23}\frac{\Phi(x_1,x_2,x_3)+\Phi(x_1,x_3,x_2)}{x_2-x_3}
+ k_{13}\frac{\Phi(x_1,x_2,x_3)+\Phi(x_3,x_2,x_1)}{x_1-x_3}.
\end{align}
This gives the odd (antisymmetric) Calogero--type operator.

\medskip

\noindent\textbf{(iii) Two-dimensional irreducible representation.}\\[2pt]
Here $\Phi(x)=\begin{pmatrix}\phi_1(x)\\ \phi_2(x)\end{pmatrix}$ and
\[
\rho(s_{12})=\begin{pmatrix}1&0\\0&-1\end{pmatrix},\qquad
\rho(s_{23})=\begin{pmatrix}-\tfrac{1}{2}&\tfrac{\sqrt{3}}{2}\\[4pt]\tfrac{\sqrt{3}}{2}&\tfrac{1}{2}\end{pmatrix},\qquad
\rho(s_{13})=\begin{pmatrix}-\tfrac{1}{2}&-\tfrac{\sqrt{3}}{2}\\[4pt]-\tfrac{\sqrt{3}}{2}&\tfrac{1}{2}\end{pmatrix}.
\]
Then
\[
\begin{aligned}
T_{e_1-e_2}&
\begin{pmatrix}\phi_1\\ \phi_2\end{pmatrix}
=(\partial_{x_1}-\partial_{x_2})
\begin{pmatrix}\phi_1\\ \phi_2\end{pmatrix}+ \frac{2k_{12}}{x_1-x_2}
\begin{pmatrix}
\phi_1(x_1,x_2,x_3)-\phi_1(x_2,x_1,x_3)\\[4pt]
\phi_2(x_1,x_2,x_3)+\phi_2(x_2,x_1,x_3)
\end{pmatrix} \\[5pt]
&-\frac{k_{23}}{x_2-x_3}
\begin{pmatrix}
\phi_1(x_1,x_2,x_3)+\tfrac{1}{2}\phi_1(x_1,x_3,x_2)-\tfrac{\sqrt{3}}{2}\phi_2(x_1,x_3,x_2)\\[4pt]
\phi_2(x_1,x_2,x_3)-\tfrac{\sqrt{3}}{2}\phi_1(x_1,x_3,x_2)-\tfrac{1}{2}\phi_2(x_1,x_3,x_2)
\end{pmatrix} \\[5pt]
&\quad \quad +\frac{k_{13}}{x_1-x_3}
\begin{pmatrix}
\phi_1(x_1,x_2,x_3)+\tfrac{1}{2}\phi_1(x_3,x_2,x_1)+\tfrac{\sqrt{3}}{2}\phi_2(x_3,x_2,x_1)\\[4pt]
\phi_2(x_1,x_2,x_3)+\tfrac{\sqrt{3}}{2}\phi_1(x_3,x_2,x_1)-\tfrac{1}{2}\phi_2(x_3,x_2,x_1)
\end{pmatrix}.
\end{aligned}
\]
The denominators $(x_i-x_j)^{-1}$ measure the distance to each reflection
hyperplane $x_i=x_j$, while the matrices $\rho(s_{ij})$ implement the internal
symmetry of the two-dimensional $S_3$ representation.  
In the scalar cases (i) and (ii), these matrices reduce to $\pm 1$, recovering
respectively the even and odd Calogero--Moser operators.

\subsection{Global formulation of $A_n$ Dirac-Dunkl operators.}
The global formulation \eqref{eq:DunklRho} extends naturally to arbitrary rank.
For the Coxeter system of type $A_n$, the reflection group is
\[
  W = S_{n+1},
\]
acting by permutation of coordinates on
\[
  V = \left\{ x = (x_1,\dots,x_{n+1}) \in \mathbb{R}^{n+1} \;\middle|\; 
  \sum_{i=1}^{n+1} x_i = 0 \right\} \;\cong\; \mathbb{R}^n.
\]
The root system is
\[
  R = \{ e_i - e_j \;|\; i \ne j \}, \qquad
  R^+ = \{ e_i - e_j \;|\; i < j \},
\]
with reflections $s_{ij}$ exchanging $x_i \leftrightarrow x_j$.
For a unitary representation $\rho: S_{n+1} \to U(d)$ and a multiplicity function
$k: R \to \mathbb{C}$ satisfying $k(\alpha_{ij})=k_{ij}=k_{ji}$, the
\emph{representation-dependent Dunkl operator} reads
\[
  T_\xi \Phi(x)
  = \partial_\xi \Phi(x)
  + \sum_{1\le i<j\le n+1} 
    k_{ij}\,\frac{\langle e_i - e_j, \xi\rangle}{x_i - x_j}\,
    \big(\Phi(x) - \rho(s_{ij})\,\Phi(s_{ij}x)\big),
\]
where $\Phi:V\to\mathbb{C}^d$ is a $d$-component function transforming under $\rho$.
The associated Dirac--Dunkl operator is
\[
  D_k = \sum_{r=1}^n e_r\,T_{e_r},
\]
with Clifford generators $e_r$ satisfying $e_r e_s + e_s e_r = -2\delta_{rs}I$.

\begin{itemize}
\item[(i)] \textbf{Trivial representation $\rho=\mathbf{1}$.}
Here all reflections act trivially, and
\[
  T_\xi f(x)
  = \partial_\xi f(x)
  + \sum_{i<j} k_{ij}\,\frac{\langle e_i - e_j, \xi\rangle}{x_i - x_j}
  \big(f(x) - f(s_{ij}x)\big),
\]
which recovers the classical scalar Dunkl operator of type $A_n$.
The Dirac--Dunkl operator reduces to the standard one of
Dunkl--Clifford analysis.

\item[(ii)] \textbf{Sign representation $\rho=\mathrm{sgn}$.}
In this case $\rho(s_{ij})=-1$ for all reflections,
so the difference term changes to a sum:
\[
  T_\xi f(x)
  = \partial_\xi f(x)
  + \sum_{i<j} k_{ij}\,\frac{\langle e_i - e_j, \xi\rangle}{x_i - x_j}
  \big(f(x) + f(s_{ij}x)\big).
\]
This operator acts naturally on antisymmetric functions and corresponds
to the “odd” or fermionic Calogero--Moser realization.

\item[(iii)] \textbf{Higher-dimensional representations.}
For any irreducible representation $\rho$ of $S_{n+1}$ with dimension $d>1$,
the Dunkl operator becomes \emph{matrix-valued}, acting on
$\Phi:V\to\mathbb{C}^d$.  Each reflection term couples the components according
to the reflection matrices $\rho(s_{ij})$:
\[
  T_\xi \Phi(x)
  = \partial_\xi \Phi(x)
  + \sum_{i<j} k_{ij}\,\frac{\langle e_i - e_j, \xi\rangle}{x_i - x_j}
    \big(\Phi(x) - \rho(s_{ij})\,\Phi(s_{ij}x)\big).
\]
The denominators $(x_i - x_j)^{-1}$ measure the distance to each reflection
hyperplane, while the matrices $\rho(s_{ij})$ encode internal symmetries of
the chosen representation.
\end{itemize}

\subsubsection{Explicit Calogero--type realization for type $A_n$.}
For $\xi = e_p - e_q$, the differential part is
$\partial_\xi = \partial_{x_p} - \partial_{x_q}$, and the operator takes the explicit form
\[
\begin{aligned}
T_{e_p-e_q}\Phi(x)
&= (\partial_{x_p} - \partial_{x_q})\,\Phi(x) \\[3pt]
&\quad+ \sum_{i<j} k_{ij}\,
  \frac{\langle e_i - e_j, e_p - e_q\rangle}{x_i - x_j}\,
  \big(\Phi(x) - \rho(s_{ij})\,\Phi(s_{ij}x)\big).
\end{aligned}
\]
For the scalar cases $\rho=\mathbf{1}$ or $\rho=\mathrm{sgn}$, this yields respectively
the even and odd Calogero--Moser operators. For higher-dimensional $\rho$, one obtains
multi-component Calogero models with internal $S_{n+1}$ symmetry.

\section{Conclusions}

The Dunkl operator $T_\xi$ is a deformation of the directional derivative
$\partial_\xi$ that incorporates reflection terms:
\[
  T_\xi^\rho f(x)
  = \partial_\xi f(x)
  + \sum_{\alpha\in R^+} k(\alpha)\,
    \frac{\langle \alpha,\xi\rangle}{\langle \alpha,x\rangle}\,
    \big(f(x) - \rho(s_\alpha)\,f(s_\alpha x)\big).
\]
It acts on $\mathbb{C}^d$-valued functions whenever $\rho$ has dimension $d$.
The corresponding Dirac--Dunkl operator combines these deformed derivatives
with Clifford generators $e_i$:
\[
  D_k = \sum_{i=1}^n e_i\,T_i,
  \qquad
  e_i e_j + e_j e_i = -2\delta_{ij} I,
\]
and satisfies the key relation
\[
  D_k^2 = -\sum_{i=1}^n T_i^2 = -\Delta_k,
\]
so that the Dirac--Dunkl operator provides a Clifford-linear square root
of the Dunkl Laplacian $\Delta_k$.  In the scalar case this recovers the usual
Calogero--Moser Hamiltonian, while for nontrivial $\rho$ one obtains
matrix-valued extensions.

The representation $\rho$ determines how the reflections act on the internal
degrees of freedom.  Depending on the choice of $\rho$, one obtains very
different analytic and physical structures, as summarized below:

\begin{itemize}
\item The Dunkl and Dirac--Dunkl operators admit a consistent
representation-dependent generalization valid for any Coxeter group, in
particular for all $A_n$ systems.

\item The {Calogero--Moser structure} appears only for scalar
representations of $W$, namely the trivial and sign representations.
In these cases, the Dunkl Laplacian $\Delta_k$ yields the
Calogero--Moser Hamiltonian, and the Dirac--Dunkl operator serves as its
square root.

\item For higher-dimensional $\rho$, the operators become
{matrix-valued}, mixing components through the reflection matrices
$\rho(s_\alpha)$.  The resulting models describe generalized
{spin-Calogero systems} with internal $S_{n+1}$ symmetry, but no longer
reduce to the scalar Calogero potential.

\item The Dirac--Dunkl framework thus unifies the scalar Calogero--Moser
equations, their fermionic analogues, and their matrix-valued spin
extensions under a single algebraic and geometric structure.
\end{itemize}

\smallskip

In conclusion, the presence or absence of the Calogero--Moser dynamics is
entirely determined by the representation content of the reflection group.
The Dirac--Dunkl operator provides a natural geometric mechanism to encode
these different realizations in a single unified formalism.

\end{document}